\documentclass[5p, times]{elsarticle}

\usepackage{ecrc}

\volume{00}

\firstpage{1}

\journalname{Systems \& Control Letters}

\runauth{}
\jid{procs}
\jnltitlelogo{Systems \& Control Letters}
\CopyrightLine{2024}{Published by Elsevier Ltd.}

\usepackage{amssymb}
\usepackage{amsmath}
\usepackage{amsfonts}
\usepackage{xcolor}
\usepackage{amsthm}
\usepackage{mathtools}
\usepackage{appendix}
\usepackage{bm}
\usepackage{float}
\usepackage{hyperref}
\usepackage{tikz}
\usetikzlibrary{arrows.meta,positioning,calc,fit,backgrounds}

\newtheorem{definition}{Definition}
\newtheorem{theorem}{Theorem}

\newtheorem{corollary}{Corollary}
\newtheorem{lemma}{Lemma}

\newtheorem{assumption}{Assumption}

\usepackage[figuresright]{rotating}

\begin{document}

\begin{frontmatter}


\dochead{}

\title{Predictor-Based Output-Feedback Control of Linear Systems \\ with Time-Varying Input and Measurement Delays \\ via Neural-Approximated Prediction Horizons
\tnoteref{funding}
\tnotetext[funding]{The work of L. Bhan is supported by
the U.S. Department of Energy (DOE) grant DE-SC0024386.
The work of M. Krsti\'c was funded by AFOSR grant FA9550-
23-1-0535 and NSF grant ECCS-2151525. The work of Y. Shi
is supported by Department of Energy grant DE-SC0025495,
Schmidt Sciences AI2050 Early Career Fellowship, and NSF
grant ECCS-2442689.}}

\author[eceDept]{Luke Bhan\corref{corresponding}}
\cortext[corresponding]{Corresponding Author}
\author[maeDept]{Miroslav Krsti\'c}
\author[eceDept]{Yuanyuan Shi}

\address[eceDept]{Department of Electrical and Computer Engineering, University of California San Diego, La Jolla, CA 92093-0411, USA}
\address[maeDept]{Department of Mechanical and Aerospace Engineering, University of California San Diego, La Jolla, CA 92093-0411, USA}

\begin{abstract}
Due to simplicity and strong stability guarantees, predictor feedback methods have stood as a popular approach for time delay systems since the 1950s. For time-varying delays, however, implementation requires computing a prediction horizon defined by the inverse of the delay function, which is rarely available in closed form and must be approximated. In this work, we formulate the inverse delay mapping as an operator learning problem and study predictor feedback under approximation of the prediction horizon. We propose two approaches: (i) a numerical method based on time integration of an equivalent ODE, and (ii) a data-driven method using neural operators to learn the inverse mapping. We show that both approaches achieve arbitrary approximation accuracy over compact sets, with complementary trade-offs in computational cost and scalability. Building on these approximations, we then develop an output-feedback predictor design for systems with delays in both the input and the measurement. We prove that the resulting closed-loop system is globally exponentially stable when the prediction horizon is approximated with sufficiently small error. Lastly, numerical experiments validate the proposed methods and illustrate their trade-offs between accuracy and computational efficiency. 
\end{abstract}

\end{frontmatter}

\section{Introduction}\label{sec:introduction}


Time-varying delays arise naturally in cyber-physical systems due to communication networks, computation pipelines, and input latencies. As a result, delays must be accounted for when designing control laws as even small delays (on the order of milliseconds) can degrade or destabilize closed-loop performance \cite{jin2017robust}. Hence, control of delayed systems has been studied extensively for decades, beginning with early approaches such as the Smith predictor \cite{smith1957predictor} and Artstein model reduction \cite{artstein1982linear}, and continuing through the development of modern delay-system analysis tools via Lyapunov-Krasovskii/Razuminkihin methods (see~\cite{doi:10.1137/1.9780898718645, fridman2014introduction,gu2003stability,9049450, 7987762, 7765078, ZHOU2016281,9760106}).

One of the most popular compensation paradigms is the class of "predictor" feedback methods, which reconstruct the future state of the system at the point where the control input will be applied \cite{Krstic2009DelayCompensation, KarafyllisKrstic2017, LHACHEMI20207677, 9536668, BEKIARISLIBERIS2024111428}. When implemented exactly, predictor-based designs can handle arbitrarily long delays, recovering-nominal closed-loop behavior. However, predictor feedback methods face several computational challenges. For nonlinear systems, the controller generally requires solving implicit ODEs \cite{zhou2014truncated, bhan2025stabilizationnonlinearsystemsunknown}. In the presence of time-varying delays, the predictor additionally depends on the inverse delay mapping to determine the time-varying prediction horizon, which is generally not available in closed form. As a result, despite their strong theoretical guarantees, predictor-based methods for non-constant delays can almost never be implemented exactly \cite{BekiarisLiberisKrstic2013}. 

In this work, we focus on the latter challenge and develop two approaches for approximating the inverse delay mapping underlying time-varying predictor feedback. The first is a classical numerical integration scheme obtained by rewriting the inverse delay function as an ODE and then solving it numerically. The second is a data-driven approach based on neural operators \cite{deeponet, FNO:2021}, which have recently gained attention in control as learned surrogates for computationally intensive numerical procedures, with speedups of up to $1000\times$. For example, they have been used to approximate backstepping gain kernels for PDE control \cite{10374221,10872816,ZHANG2026112809}, adaptive control design \cite{LAMARQUE2025112329,10918744,BHAN2025105968},  implicit ODE mappings arising in predictor constructions \cite{pmlr-v283-bhan25a, bajraktari2026delay}, and operators appearing in traffic-flow and oil-system models \cite{LV2025112553,TOUMI2025106191}. Moreover, within the same framework, we move beyond the standard state-feedback setting and present the first predictor-feedback output-feedback design for linear systems with both time-varying input and measurement delays. The resulting controller combines observer-based state reconstruction from delayed measurements with predictor-based compensation of the input delay, while retaining the implementation challenges associated with practical predictor computation.

Therefore, the main contributions of this work are summarized as:
\begin{itemize}
\item \emph{We establish an operator-theoretic framework for approximating the inverse delay mapping.} In particular, we reformulate inverse delay recovery as an operator approximation problem and prove that the inverse delay mapping is Lipschitz continuous on compact sets. This structural result provides the regularity foundation needed to analyze approximation schemes in a unified and rigorous manner.

\item \emph{Building on this framework, we derive two approximation theorems that provide rigorous guarantees for two distinct constructions of the inverse delay mapping over compact sets.} In particular, the first theorem establishes approximation guarantees for a numerical integration approach based on an equivalent ODE formulation of the inverse mapping, while the second theorem characterizes approximation by a neural operator that learns the inverse mapping over compact families of delay functions. Together, these results show that both the model-based and learning-based approaches fit naturally within the operator-theoretic framework and admit explicit approximation guarantees on compact sets.

\item \emph{We develop a predictor-based control framework, formulated through the transport PDE representation, that is robust to inverse-delay approximation error and preserves closed-loop stability under suitable conditions.} More specifically, we incorporate approximate inverse delay mappings into the predictor design and use the transport PDE formulation to track how approximation errors propagate through the predictor dynamics. This leads to sufficient conditions under which robust exponential closed-loop stability is maintained despite imperfect inversion of the delay mapping.
\end{itemize}

\color{blue}
\color{black}

\textbf{Notation:}
The Euclidean norm is denoted by $|\cdot|$, and the induced matrix norm by $\|\cdot\|$. For a symmetric matrix $M>0$, $\lambda_{\min}(M)$ and $\lambda_{\max}(M)$ denote its minimum and maximum eigenvalues. For a function $f$, $\|f\|_{L^\infty(a,b)} := \sup_{\tau\in[a,b]}|f(\tau)|$ and $\|f\|_{L^2(a,b)} := \left(\int_a^b |f(\tau)|^2 d\tau\right)^{1/2}$. For distributed states 

\noindent $u(x,t)$, $\|u(t)\|_2^2 := \int_0^1 |u(x,t)|^2 dx$. 
We use $\mathbb{R}_+ = (0,\infty)$. The space $C^k(\Omega;\mathbb{R}^q)$ denotes $k$-times continuously differentiable functions. Lastly, time derivatives are denoted by $\dot{f}(t)$ and for distributed states, $u_x(x, t)$ and $u_t(x, t)$ denote the partial space and time derivatives respectively. 

\color{blue}
A preliminary version of this work was submitted to the 2026 IEEE Conference on Decision and Control \cite{bhan2026neuralCDC}. This manuscript is substantially different and significantly extends that earlier version. Specifically, the preliminary work addressed only the state-feedback case and did not include the new control law developed in Section \ref{sec:feedback-scheme}. Moreover, the treatment of the output-feedback case in this manuscript requires a considerably more sophisticated stability analysis in Section \ref{sec:stability} together with expanded and improved numerical results in Section \ref{sec:numerics}.
\color{black}

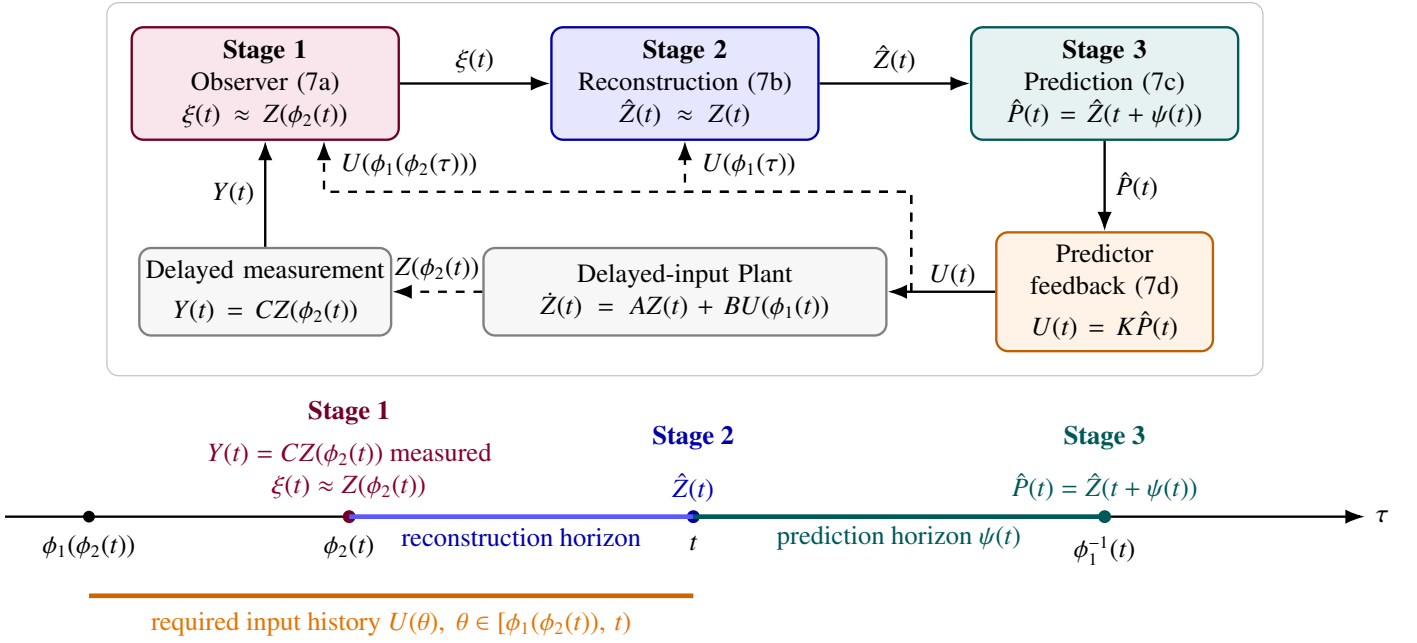
\begin{figure*}[!t]
\centering
\resizebox{\textwidth}{!}{%
\begin{tikzpicture}[
    >=Latex,
    font=\small,
    block1/.style={
        draw,
        rounded corners,
        thick,
        align=center,
        fill=purple!10,
        draw=purple!60!black,
        minimum height=1.35cm,
        text width=2.9cm,
        inner sep=4pt
    },
    block2/.style={
        draw,
        rounded corners,
        thick,
        align=center,
        fill=blue!10,
        draw=blue!65!black,
        minimum height=1.35cm,
        text width=2.9cm,
        inner sep=4pt
    },
    block3/.style={
        draw,
        rounded corners,
        thick,
        align=center,
        fill=teal!10,
        draw=teal!70!black,
        minimum height=1.35cm,
        text width=2.9cm,
        inner sep=4pt
    },
    smallblockA/.style={
        draw,
        rounded corners,
        thick,
        align=center,
        fill=gray!6,
        draw=black!50,
        minimum height=1.05cm,
        text width=3cm,
        inner sep=0pt
    },
    smallblockB/.style={
        draw,
        rounded corners,
        thick,
        align=center,
        fill=orange!10,
        draw=orange!75!black,
        minimum height=1.05cm,
        text width=2.3cm,
        inner sep=4pt
    },
    arrow/.style={-Latex, thick},
    darrow/.style={-Latex, thick, dashed}
]

\coordinate (L) at (0,-1.8);
\coordinate (R) at (16.2,-1.8);
\coordinate (C) at ($(L)!0.5!(R)$);

\def\yTop{2.0}
\def\yBottom{-0.5}
\def\yTimeline{-5}

\node[smallblockA] (meas) at ($(C)+(-5,\yBottom)$)
{Delayed measurement\\[1mm]
$Y(t)=CZ(\phi_2(t))$};

\node[block1] (obs) at ($(C)+(-5,\yTop)$)
{\textbf{Stage 1}\\
Observer \eqref{eq:control-law-1}\\
$\xi(t)\approx Z(\phi_2(t))$};

\node[block2] (recon) at ($(C)+(0,\yTop)$)
{\textbf{Stage 2}\\
Reconstruction \eqref{eq:control-law-2}\\
$\hat Z(t)\approx Z(t)$};

\node[block3] (pred) at ($(C)+(5,\yTop)$)
{\textbf{Stage 3}\\
Prediction \eqref{eq:control-law-3}\\
$\hat P(t)=\hat Z(t+\psi(t))$};

\node[smallblockB] (fb) at ($(C)+(5,\yBottom)$)
{Predictor feedback \eqref{eq:control-law-4}\\[1mm]
$U(t)=K\hat P(t)$};

\node[smallblockA, text width=4.8cm] (plant) at ($(C)+(0,\yBottom)$)
{Delayed-input Plant\\
$\dot{Z}(t)=AZ(t)+BU(\phi_1(t))$};

\draw[arrow] (meas) -- node[pos=0.5,left] {$Y(t)$} (obs);
\draw[arrow] (obs) -- node[above] {$\xi(t)$} (recon);
\draw[arrow] (recon) -- node[above] {$\hat Z(t)$} (pred);
\draw[arrow] (pred) -- node[right] {$\hat P(t)$} (fb);

\draw[arrow] (fb.west) -- ++(-1.0,0) coordinate (uTap)
    |- node[pos=0.25,above,xshift=14pt,yshift=-2pt] {$U(t)$} (plant.east);

\draw[darrow] (plant.west) -- node[above] {$Z(\phi_2(t))$} (meas.east);

\draw[darrow]
(uTap) -- ++(0,1.2)
-| node[pos=0.85,right,xshift=2pt,yshift=-3pt] {$U(\phi_1(\phi_2(\tau)))$}
($(obs.south)+(0.7,0)$);

\draw[darrow]
(uTap) -- ++(0,1.2)
-| node[pos=0.85,right,xshift=2pt,yshift=-3pt] {$U(\phi_1(\tau))$}
(recon.south);

\draw[thick,-Latex] (0,\yTimeline) -- (16.2,\yTimeline) node[right] {$\tau$};

\coordinate (a) at (1.0,\yTimeline);
\coordinate (b) at (4.1,\yTimeline);
\coordinate (c) at (8.2,\yTimeline);
\coordinate (d) at (13.1,\yTimeline);

\fill[black!80!black] (a) circle (2pt);
\fill[purple!60!black] (b) circle (2.2pt);
\fill[blue!65!black] (c) circle (2.2pt);
\fill[teal!70!black] (d) circle (2.2pt);

\node[text=purple!70!black, align=center] at ($(b)+(0,0.55)$)
{$Y(t)=CZ(\phi_2(t))$ measured\\ $\xi(t)\approx Z(\phi_2(t))$};

\node[text=blue!70!black] at ($(c)+(0,0.35)$)
{$\hat Z(t)$};

\node[text=teal!70!black, align=center] at ($(d)+(0,0.35)$)
{$\hat P(t)=\hat Z(t+\psi(t))$};

\node[text=purple!60!black] at ($(b)+(0,1.25)$) {\textbf{Stage 1}};
\node[text=blue!65!black] at ($(c)+(0,0.95)$) {\textbf{Stage 2}};
\node[text=teal!70!black] at ($(d)+(0,0.95)$) {\textbf{Stage 3}};

\draw[ultra thick, blue!65] (b) -- (c);
\draw[ultra thick, teal!70!black] (c) -- (d);

\node[text=blue!70!black] at ($(b)!0.5!(c)+(0,-0.25)$)
{reconstruction horizon};

\node[text=teal!70!black] at ($(c)!0.5!(d)+(0,-0.25)$)
{prediction horizon $\psi(t)$};

\node[below=3pt] at (a) {$\phi_1(\phi_2(t))$};
\node[below=3pt] at (b) {$\phi_2(t)$};
\node[below=3pt] at (c) {$t$};
\node[below=3pt] at (d) {$\phi_1^{-1}(t)$};

\draw[ultra thick, orange!80!black] ($(a)+(0,-0.95)$) -- ($(c)+(0,-0.95)$);
\node[below=2pt, text=orange!90!black] at ($(a)!0.5!(c)+(0,-0.95)$)
{required input history $U(\theta),\ \theta\in[\phi_1(\phi_2(t)),\,t)$};

\begin{scope}[on background layer]
\node[
    draw=black!25,
    rounded corners,
    inner sep=8pt,
    fit=(meas)(obs)(recon)(pred)(fb)(plant)
] {};
\end{scope}

\end{tikzpicture}%
}
\caption{Three-stage observer--predictor feedback law and corresponding horizons.}
\label{fig:observer_predictor_feedback}
\end{figure*}
\section{Output-feedback Control Design}\label{sec:feedback-scheme}
In this work, we study:
\begin{subequations}
\begin{align}
    \dot{Z}(t) =\;& AZ(t) + BU(t-D_1(t))\,, \label{eq:dynamics-1} \quad t \geq 0\\
    Y(t) =\;& CZ(t-D_2(t))\,, \label{eq:dynamics-2} \quad t \geq 0
\end{align}
\end{subequations}
where $Z \in \mathbb{R}^n$ is the state, $U\in \mathbb{R}^m$ is the control input, $D_1, D_2: \mathbb{R} \to \mathbb{R}_+$ are delay functions, and $Y(t)$ is the system measurement available to the engineer. We begin by establishing the two preliminary function definitions:
\begin{equation}
\phi_1(t) = t-D_1(t)\,, \quad \phi_2(t) = t-D_2(t)\,. \label{eq:delay-func-defs}
\end{equation}
$\phi_1$ represents the delay-time, that is the current delay time stamp at $t$ in the input. $\phi_2$ gives the delayed-state measurement at time $t$ - that is how far in the past we are measuring the state.

To ensure well-posedness of the system, we make the following assumptions regarding the delay functions in \eqref{eq:delay-func-defs}

\begin{assumption} \label{assumption:delay-time-positive-and-uniformly-bounded}
Let $\phi_i:\mathbb{R}_+ \to \mathbb{R}$, $i\in\{1,2\}$, be continuously differentiable functions. Assume there exists a constant $\pi_0^\ast>0$ such that
\begin{align}
    t-\phi_i(t)\ge \pi_0^\ast>0\,, \quad \forall t>0,\; i\in\{1,2\}.
\end{align}
Moreover, define
\begin{align}
    \pi_1^\ast
    :=
    \frac{1}{
    \displaystyle
    \sup_{i\in\{1,2\}}
    \sup_{\nu\ge \phi_i^{-1}(0)}
    (\nu-\phi_i(\nu))
    }>0.
\end{align}
\end{assumption}

\begin{assumption} \label{assumption:invertibility}
Let $\phi_i:\mathbb{R}_+ \to \mathbb{R}$, $i\in\{1,2\}$, be continuously differentiable functions. Assume there exists a constant $\pi_2^\ast>0$ such that
\begin{align}
    \phi_i'(t)\ge \pi_2^\ast>0\,, \quad \forall t>0,\; i\in\{1,2\}.
\end{align}
Moreover, define
\begin{align}
    \pi_3^\ast
    :=
    \frac{1}{
    \displaystyle
    \sup_{i\in\{1,2\}}
    \sup_{\nu\ge \phi_i^{-1}(0)}
    \phi_i'(\nu)
    }>0.
\end{align}
\end{assumption}

Note that both assumptions are required for the \emph{exact} implementation of predictor feedback under time-varying delays \cite{5373886, 10.1115/1.4005278}. Specifically, they ensure that $(a)$ the delay remains strictly positive and $(b)$ its temporal variation is sufficiently slow, in the sense that $\dot{D}(t)<1$. Thus, the analysis pertains to systems with \emph{slowly varying delay mechanisms}, including transport-dominated processes, network-induced delays with mild time variation \cite{LIU201957}, and sensing or actuation pipelines whose latency changes continuously with the operating condition.

We propose the following control design in three segments:
\begin{enumerate}
    \item Construct an observer for the delayed state $Z(\phi_2(t))$ using the measurement $Y(t)$, yielding the estimate $\xi(t)$:
\begin{subequations}
\begin{align}
    \dot{\xi}(t) =&\; \phi'_2(t)\Bigl[A\xi(t) + BU(\phi_1(\phi_2(t))) + L\bigl(Y(t) - C\xi(t)\bigr)\Bigr], \label{eq:control-law-1} \\
\intertext{    \item Propagate this estimate forward in time using the plant dynamics to obtain $\hat{Z}(t)$, an estimate of the current state:
}
    \hat{Z}(t) =&\; e^{A(t-\phi_2(t))}\xi(t) + \int_{\phi_2(t)}^t e^{A(t-\tau)}B\,U(\phi_1(\tau))\, d\tau, \label{eq:control-law-2} \\ 
    \intertext{    \item Propagate the estimate further to the predicted future state $\hat{P}(t) = \hat{Z}(t+ \psi(t))$, where $\psi(t) = \phi^{-1}_1(t)-t$ is the prediction horizon. We then use this predicted state as input to the nominal control law:
}
\hat{P}(t)
=&\; e^{A\psi(t)}\,\hat{Z}(t) \nonumber \\
&+ \int_{t}^{t+\psi(t)}
e^{A\big(t+\psi(t)-s\big)}
\,B\,U\big(\phi_1(s)\big)\,ds,
\label{eq:control-law-3}\\
U(t) =&\; K\hat{P}(t).
\label{eq:control-law-4}
\end{align}
\end{subequations}
\end{enumerate}

Notice that the control law requires an initialization of both the control and the state measurements as 
\begin{subequations}
    \begin{align}
    U(\theta)\,, &\quad \theta \in[\phi_1(\phi_2(0)), 0]\,, \\
    Z(\theta)\,, &\quad \theta \in [\phi_2(0), 0]\,, 
\end{align}
\end{subequations}
with the observer initial condition at $t=0$ as $\xi(0) = \xi_0 \in \mathbb{R}^n$. A full diagram of the stages of the control law, as well as the various horizons in this work is given in Figure \ref{fig:observer_predictor_feedback}. Furthermore, we choose $K$ and $L$ such that both $A+BK$ and $A-LC$ are Hurwitz as is standard in linear systems. In this work, we will show that the exact implementation of this scheme satisfies the following theorem:

\begin{theorem} \label{thm:exact-theorem}
    Consider the plant \eqref{eq:dynamics-1}, \eqref{eq:dynamics-2} such that Assumptions \ref{assumption:delay-time-positive-and-uniformly-bounded} and \ref{assumption:invertibility} are satisfied. Then, there exist constants $M, C> 0$ such that all solutions of the plant with the \textbf{exact} feedback control law \eqref{eq:control-law-1}-\eqref{eq:control-law-4} satisfy for all $t \geq 0$
    \begin{align}
        \Gamma(t) \leq&\; Me^{-Ct} \Gamma(0)\,, \\
    \Gamma(t) :=&\; |Z(t)|^2 + |Z(t)-\hat{Z}(t)|^2 + \sup_{\phi_1(t)\le \tau \le t} |U(\tau)|^2.
    \end{align}
\end{theorem}

Theorem \ref{thm:exact-theorem} claims global exponential stability of the full feedback control law and we will see that Theorem \ref{thm:exact-theorem} is a special case of our main result. However, there is one major issue with the theorem in that it assumes \eqref{eq:control-law-1}-\eqref{eq:control-law-4} are implemented exactly. For most delay functions, the control law \eqref{eq:control-law-1}-\eqref{eq:control-law-4} cannot be computed exactly as the predictor $\hat{P}(t)$ requires the inverse of the delay time $\psi(t) = \phi_1^{-1}(t)-t$. Hence, this motivates developing the following two approximation schemes for $\psi$ in Section \ref{sec:approximation-schemes}

\section{Approximation schemes for the delay horizon time} \label{sec:approximation-schemes}
As we saw in Section \ref{sec:feedback-scheme}, the analytical inverse of the delay-time function is rarely computable. Consequently, in this section, we will propose two approaches for its approximation. We begin by defining an operator formulation of the mappings:
\begin{definition}\label{definition:delay-time}
    (\textbf{Delay time operator}) Consider any function $D(t) \in C^1(\mathbb{R}_+; \mathbb{R}_+)$. Then, define the delay time operator $\Phi: C^1(\mathbb{R}_+, \mathbb{R}_+) \to C^1(\mathbb{R}_+, \mathbb{R}_+)$ as the affine mapping satisfying
    \begin{align}
        \Phi(D)(t) := t-D(t) \,, \quad t \in \mathbb{R}_+.
    \end{align}
\end{definition}

\begin{definition} \label{definition:inverse-delay-time}
    (\textbf{Prediction horizon time operator})  Consider any function $D(t) \in C^1(\mathbb{R}_+; \mathbb{R}_+)$ and corresponding  $\phi(t) = t- D(t)$ satisfying  Assumptions \ref{assumption:delay-time-positive-and-uniformly-bounded} and \ref{assumption:invertibility}. The corresponding inverse function is implicitly given such that $\phi^{-1}(\phi(t)) = t$. Then, define the inverse delay time operator $\Psi: C^1(\mathbb{R}_+, \mathbb{R}_+) \to C^1(\mathbb{R}_+, \mathbb{R}_+)$ as
    \begin{align}
        \Psi(D)(t):= \phi^{-1}(t)-t\,,
    \end{align}
    where $\phi^{-1}(\phi(t)) = t$.
\end{definition}

For convenience, we denote the operator output by
\begin{align}
    \psi(t) = \Psi(D)(t).
\end{align}
The first approach is a numerical time-stepping scheme obtained by approximating the ODE characterizing $\Psi$ using Euler's method. We begin by establishing a key property of the operator $\Psi$, namely that it is Lipschitz.

\begin{lemma}[Lipschitz dependence of $\Psi$ on $D$] \label{lem:inverse_lipschitz_in_D}
Let $D_1,D_2 \in C^1(\mathbb{R}_{> 0};\mathbb{R}_{> 0})$ and define
$\phi_i(t) := t - D_i(t), \quad i\in\{1,2\}.$
Assume that each $\phi_i$ satisfies Assumptions \ref{assumption:delay-time-positive-and-uniformly-bounded} and \ref{assumption:invertibility}. 
Then, for every $y$ in the common range $\mathcal{R}\coloneq \phi_1(\mathbb{R}_{> 0})\cap \phi_2(\mathbb{R}_{> 0})$,
\begin{align}
\bigl|\Psi(D_1)(y)-\Psi(D_2)(y)\bigr|
\;\leq\;
\frac{1}{\pi_2^\ast}\,\|D_1-D_2\|_{L^\infty(\mathbb{R}_{> 0})}.
\end{align}
where $\Psi$ is given as in Definition \ref{definition:inverse-delay-time}. 
\end{lemma}

The proof is given in Appendix \ref{lem:inverse_lipschitz_in_D}. 

\subsection{Euler-approximation of the inverse delay time}
The Euler method for approximating inverses has been well studied, but never applied in the context of predictor feedback problems \cite{atkinson2009numerical}, \cite{KarafyllisKrstic2017}. Hence, we provide a terse explanation of the approach and corresponding error theorem to highlight the challenges. 
From the identity $\phi(\psi(t)+t) = t$, the derivative of $\psi(t)$ satisfies the ODE
\begin{align} \label{eq:psi-ode}
    \frac{d \psi(t)}{dt} = \frac{\dot{D}(t+\psi(t))}{1-\dot{D}(t+\psi(t))} \,, 
\end{align}
where the initial condition satisfies the fixed-point equation
$\psi(0) - D(\psi(0)) = 0$.
This equation is well-posed under Assumption \ref{assumption:invertibility}. 
The ratio appearing in \eqref{eq:psi-ode} reflects the rate of change of the delay through $\dot{D}$. 
When $\dot{D}(t)$ approaches $1$, the prediction horizon varies rapidly, whereas when the delay varies slowly, the horizon length remains nearly constant. 
To compute $\psi(t)$, one first solves the fixed-point equation for $\psi(0)$ and then numerically integrates the ODE forward in time. 
Under the classical Euler approximation of the ODE, one obtains the following well-known theorem:
\begin{theorem}[Explicit Euler approximation of $\psi$ \cite{atkinson2009numerical}]
\label{thm:euler_phi_inverse_explicit}
Suppose Assumptions~\ref{assumption:delay-time-positive-and-uniformly-bounded} and \ref{assumption:invertibility} hold and, in addition, $\phi \in C^2(\mathbb{R}_{>0})$. Then consider the differential equation satisfying
\begin{subequations}
\begin{align}
    \dot{\psi}(t)
    &= \frac{\dot D(t+\psi(t))}{1-\dot D(t+\psi(t))}, \\
    \psi(0) &= t_0, \label{eq:euler-init-cond-1} \\
    t_0-D(t_0) &= 0.  \label{eq:euler-init-cond-2}
\end{align}
\end{subequations}
Let $h=T/N$ and define $t_n=nh$ for $n=0,\dots,N$. Define the explicit Euler approximation $\hat{\psi}(t_n)$ by
\begin{align}
    \hat{\psi}(t_{n+1})=\hat{\psi}(t_n)
    + h\frac{\dot D(t_n+\hat \psi(t_n))}{1-\dot D(t_n+\hat \psi(t_n))}
\end{align}
Then, on any compact interval $[0, T]$, with $T> 0$, the explicit Euler approximation satisfies
\begin{align}
    \max_{0\leq t_n \leq T}
    \left|
    \psi(t_n)-\hat{\psi}(t_n)
    \right|
    \leq
    \left[
    \frac{e^{KT}-1}{K}
    \right]
    \frac{1}{2}h
    \max_{0\leq t\leq T}
    \left|
    \ddot{\psi}(t)
    \right|,
\end{align}
where the Lipschitz constant $K$ may be chosen as
\begin{align}
    K
    &=
    \sup_{0\le t\le T}
    \left|
        \frac{\phi''(t+\psi(t))}{(\phi'(t+\psi(t)))^2}
    \right|
   \\& \le
    \frac{1}{(\pi_3^\ast)^2}
    \sup_{0\le t\le T}
    \left|
        \phi''(t+\psi(t))
    \right|.
\end{align}
\end{theorem}

Notice that, to guarantee a uniform error bound of size $\epsilon>0$ on $[0,T]$, it suffices to choose the step-size $h$ so that
\begin{align}
    h
    \le
    \frac{2\epsilon K}
    {\left(e^{KT}-1\right)
    \displaystyle\max_{0\le t\le T}\left|\ddot{\psi}(t)\right| }.
\end{align}

To interpret the constants in the Euler error bound, note that the derivatives of $\phi$ and $\psi$ are determined directly by the delay function $D(t)$. Since $\phi(t)=t-D(t)$, we have $\phi''(t)=-\ddot D(t)$. Moreover, recalling that $\psi(t)=\phi^{-1}(t)-t$, one obtains $\ddot{\psi}(t)=\ddot D(t+\psi(t))/(1-\dot D(t+\psi(t)))^3$. Thus the quantities $\phi''$ and $\ddot{\psi}$ appearing in the error bound reflect the curvature of the delay profile $D(t)$ along the prediction horizon. Intuitively, slower varying delays (small $\dot D$ and $\ddot D$) lead to smaller bounds in the Euler approximation error.

The explicit Euler scheme is presented only as a simple illustrative example; higher-order integration methods (e.g., Runge--Kutta schemes) could be employed to improve the local truncation error. However, such methods have two drawbacks. First, for all the stability analysis that follows, we only require $D \in C^1$ and hence, higher-order numerical methods will require stronger regularity in the delay function. Second, the qualitative dependence of the global error bound on the time horizon $T$ is not specific to Euler. For non-contractive dynamics, time-marching schemes typically yield estimates whose constants grow with $T$, as a consequence of Gr\"onwall-type arguments \cite{hairer1993ode1}. 

Thus, although any prescribed tolerance $\epsilon > 0$ can in principle be achieved by selecting $h$ sufficiently small, the required step size generally deteriorates as $T$ increases. In other words, maintaining uniform accuracy over long horizons becomes increasingly expensive. This observation motivates approximating the inverse mapping $\psi$ directly as an operator, rather than computing it via sequential time integration.

Lastly, note that this growth with respect to $T$ does not appear in the predictor approximation schemes of \cite{KarafyllisKrstic2017}. There, the predictor evolves over a fixed interval equal to the delay length, so the approximation error depends only on the known delay magnitude rather than on an arbitrarily long integration horizon. This illustrates a key difference in approximating feedback laws under time-varying delays.

\subsection{Neural operator approximation of the inverse delay time}
Further, we also present a neural operator approach to approximating $\psi(t)$. For completeness, we begin with a brief review of operator learning. Neural operators are neural networks with structured kernel parameterizations that approximate mappings between function spaces, admitting universal operator approximation guarantees (see \cite{Neural_Operator:2023}). Prominent architectures include DeepONet \cite{deeponet}, the Fourier Neural Operator (FNO) \cite{FNO:2021}, and transformer-based variants \cite{Neural_Operator:2023}. We provide the formal definition of a neural operator as

\begin{definition}[Neural Operators] \label{definition:neural-operator} \cite[Section 1.2]{LLS2024} Let $\Omega_u \subset \mathbb{R}^{d_{u_1}}$, $\Omega_v \subset \mathbb{R}^{d_{v_1}}$ be bounded domains with Lipschitz boundary and let  $\mathcal{F}_c \subset C^0(\Omega_u; \mathbb{R}^c)$, $\mathcal{F}_v \subset C^0(\Omega_v; \mathbb{R}^v)$ be continuous function spaces.
Given a channel dimension $d_c > 0$, we call any $\hat{\Psi}$ a neural operator given it satisfies the compositional form $\hat{\Psi} = \mathcal{Q} \circ \mathcal{L}_L \circ \cdots \circ \mathcal{L}_1 \circ \mathcal{R}$ where  $\mathcal{R}$ is a lifting layer, $\mathcal{L}_l, l=1,..., L$ are the hidden layers, and $\mathcal{Q}$ is a projection layer. That is, $\mathcal{R}$ is given by
{
\small
\begin{equation}
    \mathcal{R} : \mathcal{F}_c(\Omega_u; \mathbb{R}^c) \rightarrow \mathcal{F}_s(\Omega_s; \mathbb{R}^{d_c}), \ c(x) \mapsto R(c(x), x), 
\end{equation}
}
where $\Omega_s \subset \mathbb{R}^{d_{s_1}}$, $\mathcal{F}_s(\Omega_s; \mathbb{R}^{d_c})$ is a Banach space for the hidden layers and $R: \mathbb{R}^c \times \Omega_u \rightarrow \mathbb{R}^{d_c}$ is a learnable neural network acting between finite-dimensional Euclidean spaces. For $l=1, ..., L$, each hidden layer is given by
{
\small
\begin{equation} \label{eq:generalNeuralOperator}
    (\mathcal{L}_l v)(x) := s \left( W_l v(x) + b_l + (\mathcal{K}_lv)(x)\right)\,, 
\end{equation}
}
where weights $W_l \in \mathbb{R}^{d_c \times d_c}$ and biases $b_l \in \mathbb{R}^{d_c}$ are learnable parameters, $s: \mathbb{R} \rightarrow \mathbb{R}$ is a smooth, infinitely differentiable activation function that acts component-wise on inputs, and $\mathcal{K}_l$ is the nonlocal operator given by
{
\small
\begin{equation} \label{eq:generalKernel}
    (\mathcal{K}_lv)(x) = \int_\mathcal{X} K_l(x, y) v(y) dy\,,
\end{equation}
}
where $K_l(x, y)$ is a kernel function containing learnable parameters. Lastly, the projection layer $\mathcal{Q}$ is given by
{
\small
\begin{equation}
    \mathcal{Q} : \mathcal{F}_s(\Omega_s; \mathbb{R}^{d_c}) \rightarrow \mathcal{F}_v(\Omega_v; \mathbb{R}^v), s(x) \mapsto Q(s(x), y), 
\end{equation}
}
where $Q$ is a finite dimensional neural network from $\mathbb{R}^{d_c} \times \Omega_v \rightarrow \mathbb{R}^v$. 
\end{definition}

Henceforth, we refer to any neural operator as any neural network satisfying the form of Definition \ref{definition:neural-operator}. Given the continuity of $\Psi$ as in Lemma \ref{lem:inverse_lipschitz_in_D}, we have the following Theorem as a consequence of \cite[Theorem 2.1]{LLS2024}
\begin{theorem}
    \label{thm:neural-op-approximate-theorem}
    Let Assumptions \ref{assumption:delay-time-positive-and-uniformly-bounded} and \ref{assumption:invertibility} hold. Fix a compact set $K \subset C^1(\Omega_u, \mathbb{R}_+)$ where $\Omega_u \subset \mathbb{R}_+$ is bounded. Then, for any $\epsilon > 0$ there exists a neural operator in the form of Definition \ref{definition:neural-operator} such that for all $t \in \Omega_u$
    \begin{align}
        \sup_{D \in K}|\Psi(D)(t) - \hat{\Psi}(D)(t)| \leq \epsilon\,. 
    \end{align}
\end{theorem}

Although the neural operator framework provides a theoretical $\epsilon$-approximation guarantee, practical challenges remain for large time horizons $T$. Accurately approximating the inverse delay over long intervals generally requires extensive training data and sufficiently expressive network architectures, so in practice the complexity may still scale with $T$ as in time-stepping schemes \cite{lanthaler2025discretizationerrorfourierneural}.

\section{Exponential stability of the closed loop design} \label{sec:stability}
We will use $\hat{\psi}(t) = \hat{\Psi}(D)(t)$ to represent the approximation function obtained by either approach and $\psi(t)$ to be the exact inverse delay time solution.
We now present our main result:
\begin{theorem} \label{thm:main-theorem}
Consider the plant \eqref{eq:dynamics-1}, \eqref{eq:dynamics-2} and suppose that Assumptions \ref{assumption:delay-time-positive-and-uniformly-bounded} and \ref{assumption:invertibility} hold. 
Consider the control law \eqref{eq:control-law-1}--\eqref{eq:control-law-4} implemented using an approximate prediction horizon function $\hat{\psi}(t)$. 
Then, there exists $\epsilon^\ast>0$ such that for any approximation satisfying
\begin{align}
    |\psi(t)-\hat{\psi}(t)| \le \epsilon, \qquad \forall t\ge0, \label{eq:approximation-cond}
\end{align}
with $\epsilon\in(0,\epsilon^\ast)$, there exists a constant $M>0$ and monotonically decreasing function $C(\epsilon)>0$, such that all solutions of the closed-loop system satisfy
\begin{align} \label{eq:main-global}
    \Gamma(t) \le M e^{-C(\epsilon)t}\Gamma(0), \qquad \forall t\ge0,
\end{align}
where
\begin{align}
    \Gamma(t) := |Z(t)|^2 + |Z(t)-\hat{Z}(t)|^2 + \sup_{\phi_1(t)\le \tau \le t} |U(\tau)|^2.
\end{align}
\end{theorem}

We briefly discuss the implications of implementing Theorem \ref{thm:main-theorem} under the proposed designs. The approximation condition \eqref{eq:approximation-cond} requires a uniform-in-time bound on the prediction horizon error. However, the approximation guarantees provided by Theorems \ref{thm:euler_phi_inverse_explicit} and \ref{thm:neural-op-approximate-theorem} are inherently finite-horizon as they require compact sets. Nevertheless, a global implementation can be achieved by choosing a compact interval for $\hat{\psi}$ large enough. Further, since the constants 
$M$ and $C$ in Theorem \ref{thm:main-theorem} depend only on the approximation error bound and system parameters, they can be chosen uniformly across intervals. This allows the stability estimate to be propagated over successive compact time intervals given that the interval length is chosen sufficiently large to maintain a contraction:

\begin{corollary}
\label{cor:reinit}
Suppose the hypotheses of Theorem~\ref{thm:main-theorem} hold, and fix
$\epsilon\in(0,\epsilon^*)$, where $\epsilon^*$ is given by
Theorem \ref{thm:main-theorem}. Let $M$ and $C:=C(\epsilon)$ be the
corresponding constants in \eqref{eq:main-global}. Then, if  $H$ is given by 
\begin{align}
    H > \frac{\ln(M)}{C(\epsilon^\ast)}
\end{align}
then, for all $\epsilon < \epsilon^\ast$ there exists an approximation $\hat{\psi}_t$ such that
\begin{align}
    |\psi(\tau) - \hat{\psi}_t(\tau)| \leq \epsilon \,, \quad \tau \in[t, t+H]\,,
\end{align}
and all solutions of the closed loop system satisfy 
\begin{align}
\Gamma(t)\le&\; M e^{-\bar C t}\Gamma(0),\qquad \forall t\ge 0\,, 
\\ \bar{C} :=&\; C(\epsilon) - \ln(M)/H\,. 
\end{align}
\end{corollary}


\begin{proof} (Of Theorem \ref{thm:main-theorem})
The proof of the result can be broken up into a series of steps. First, we introduce the following backstepping transformation which, due to the $\psi$ approximation, will have a boundary term. Then, we will bound this boundary term by the state estimation error.

 \noindent  \underline{\textbf{Step 1: An estimate on the target system}}
  We begin by rewriting the delay ODE as the following transport PDE-ODE cascade by introducing the distributed input
\begin{align}
    u(x, t) = U(\phi_1(t+x\psi(t)))\,, 
\end{align}
which yields the system
\begin{subequations}
\begin{align}
    \dot{Z}(t) =\;& AZ(t) + Bu(0, t)\,, \\
    u_t(x, t) =\;& \pi(x, t) u_x(x, t)\,, \quad x \in(0, 1)\\
    u(1, t) =\;& U(t)\,,\\
    Y(t) =\;& CZ(\phi_2(t))
\end{align}
\end{subequations}
where the transport speed $\pi(x, t)$ is given by 
\begin{align}
    \pi(x, t) = \frac{1 + x\dot{\psi}(t)}{\psi(t)} 
\end{align}
Note, Assumption \ref{assumption:delay-time-positive-and-uniformly-bounded} guarantees that the denominator is non-zero and hence the transport speed cannot blow up.

Consider the invertible backstepping transform \cite{5373886} governed by
\begin{align}
    w(x, t) = u(x, t) - K e^{Ax\psi(t)} Z(t) - K\int_0^x e^{A(x-y)\psi(t)}Bu(y, t) \psi(t) dy\,. 
\end{align}
This yields the following target system with boundary perturbation
\begin{subequations}
\begin{align}
    \dot{Z}(t) =&\; (A+BK)Z(t) + Bw(0, t)\,, \\ 
    w_t(x, t) =&\; \pi(x, t) w_x(x, t) \,, \quad x \in (0, 1)\,, \\
    w(1, t)     =&\; K[e^{A\hat{\psi}(t)}\hat{Z}(t)-e^{A\psi(t)}Z(t)] \\ &+ K \int_0^1 Bu(y, t)[\hat{\psi}(t)e^{A(1-y)\hat{\psi}(t)} - \psi(t)e^{A(1-y)\psi(t)}] dy 
\end{align}
\end{subequations}

We begin by analyzing the $w$ part of the cascaded system. Consider the Lyapunov function
\begin{align}
    L(t) = \frac{1}{2} \int_0^1 e^{bx} w^2(x, t) dx\,, 
\end{align}
where $b > 0$ is a constant. Computing the time derivative yields
\begin{align}
    \dot{L}(t) =& \int_0^1 e^{bx} w(x, t) w_t(x, t) dx \nonumber \\ 
    =& \int_0^1 e^{bx} w(x, t) \pi(x, t) w_x(x, t)dx \nonumber \\ 
    =& \frac{e^{bx}}{2} \pi(x, t) w^2(x, t) |_0^1 \nonumber \\ 
    &- \frac{1}{2} \int_0^1 (b\pi(x, t) + \pi_x(x, t))e^{bx}w^2(x, t) dx \nonumber \\ 
    =& -\frac{\pi(0, t)}{2} w^2(0, t)    + \frac{\pi(1, t)e^b }{2} w^2(1, t)\nonumber \\ &- \frac{1}{2} \int_0^1 (b\pi(x, t) + \pi_x(x, t))e^{bx}w^2(x, t) dx 
\end{align}
One can show that by choosing $b \geq (1-\pi_3^\ast)\max\left\{1, \frac{1}{\pi_3^\ast}\right\}$, then one obtains
\begin{align}
    b\pi(x, t) + \pi_x(x, t) \geq \pi_1^\ast \beta^\ast \,,
\end{align}
where
\begin{align}
    \beta^\ast = \min\{b-1 + \pi_3^\ast, (b+1)\pi_3^\ast - 1\} > 0\,.
\end{align}
Further, we have 
\begin{align}
    \pi(0, t) =&\; 1/\psi(t) \geq \pi_1^\ast, \\
    \pi(1, t) =&\; \frac{1+\dot{\psi}(t)}{\psi(t)} \leq \frac{1}{\pi_0^\ast\pi_2^\ast}\,,
\end{align}
and hence we obtain
\begin{align}
    \dot{L}(t) \leq -\frac{\pi_1^\ast}{2} w^2(0, t) - \pi_1^\ast \beta^\ast L(t)  + \frac{e^b}{2\pi_0^\ast\pi_2^\ast } w^2(1,t) \,. 
\end{align}

Hence, the only term that creates challenges is the boundary term $w(1, t)$.

\noindent \underline{\textbf{Step 2: A bound on $w(1, t)$.}}
We begin with the first term in the boundary $w(1, t)$. Adding and subtracting $e^{A\psi(t)}\hat{Z}(t)$ and $e^{A\hat{\psi}(t)}Z(t)$ gives
\begin{align}
K\bigl[e^{A\hat{\psi}(t)}\hat{Z}(t)-e^{A\psi(t)}Z(t)\bigr]
=&\; K e^{A\psi(t)}\bigl(\hat{Z}(t)-Z(t)\bigr) \nonumber \\
&+ K\bigl(e^{A\hat{\psi}(t)}-e^{A\psi(t)}\bigr)\bigl(\hat{Z}(t)-Z(t)\bigr) \nonumber \\
&+ K\bigl(e^{A\hat{\psi}(t)}-e^{A\psi(t)}\bigr)Z(t).
\end{align}
Hence,
\begin{align}
\left|K\bigl[e^{A\hat{\psi}(t)}\hat{Z}(t)-e^{A\psi(t)}Z(t)\bigr]\right|
\leq&\; \|K\|\,\|e^{A\psi(t)}\|\,|\hat{Z}(t)-Z(t)| \nonumber \\
&+ \|K\|\,\|e^{A\hat{\psi}(t)}-e^{A\psi(t)}\|\,|\hat{Z}(t)-Z(t)| \nonumber \\
&+ \|K\|\,\|e^{A\hat{\psi}(t)}-e^{A\psi(t)}\|\,|Z(t)|. 
\end{align}
Using Assumption \eqref{assumption:delay-time-positive-and-uniformly-bounded} and the fact $|\hat{\psi}(t) - \psi(t)| \leq  \epsilon \quad \forall t$\,, we have 
\begin{align}
    \|K\|\|e^{A\psi(t)}\| &\leq \|K\|e^{\|A\|/\pi_1^\ast} =: \Omega_1 
\end{align}
and 
\begin{align}
    \|K\| \|e^{A\hat{\psi}(t)} - e^{A\psi(t)}\| =&\; \|K\| \|e^{A\psi(t)}\|\|e^{A[\hat{\psi}(t)-\psi(t)]} -I \|  \nonumber \\ 
    \leq&\;  \Omega_1(e^{\|A\|\epsilon}-1) =: \delta_1(\epsilon)\,. 
\end{align} 
Thus, we have
\begin{align}
    \left|K\bigl[e^{A\hat{\psi}(t)}\hat{Z}(t)-e^{A\psi(t)}Z(t)\bigr] \right| \leq&\; \left[\Omega_1+\delta_1(\epsilon)\right]|\hat{Z}(t) - Z(t)| \nonumber \\ 
    &\; + \delta_1(\epsilon) |Z(t)|\,.  \label{eq:1st-term-estimate}
\end{align}
We now focus on the second term:

\begin{align}
\big|
K \int_0^1 &Bu(y,t)\big[\hat{\psi}(t)e^{A(1-y)\hat{\psi}(t)}-\psi(t)e^{A(1-y)\psi(t)}\big]\,dy
\big|
\nonumber\\
\le{}\;&
\|K\|\|B\|\int_0^1 |u(y,t)|
\Big(
|\hat{\psi}(t)-\psi(t)|\,\|e^{A(1-y)\hat{\psi}(t)}\|
\nonumber \\ &+
|\psi(t)|\,\|e^{A(1-y)\hat{\psi}(t)}-e^{A(1-y)\psi(t)}\|
\Big)\,dy
\nonumber\\
\le{}\;&
\|B\|
\left(
\Omega_1|\hat{\psi}(t)-\psi(t)|
+
|\psi(t)|\delta_1(\epsilon)
\right)\int_0^1 |u(y,t)|\,dy
\nonumber\\
\le{}\;&
\|B\|\big(\Omega_1\epsilon+|\psi(t)|\delta_1(\epsilon)\big)\int_0^1 |u(y,t)|\,dy . \label{eq:2nd-term-estimate}
\end{align}
Now, using the fact $|\psi(t)| \leq \frac{1}{\pi_1^\ast}$, along with estimate in \eqref{eq:1st-term-estimate}, \eqref{eq:2nd-term-estimate} 
\begin{align}
    |w(1, t)| \leq&\; |\Omega_1+\delta_1(\epsilon)||\hat{Z}(t) - Z(t)| + \nonumber  \delta_1(\epsilon)|Z(t)| \\ &\;+ \|B\|\left(\Omega_1\epsilon + \frac{\delta_1(\epsilon)}{\pi_1^\ast} \right) \|u(t)\|_2. 
\end{align}
Using the inequality $(a+b+c)^2 \leq 3a^2+3b^2+3c^2$, we have 
\begin{align}
    w^2(1, t) \leq&\; 3|\Omega_1+\delta_1(\epsilon)|^2|\hat{Z}(t) - Z(t)|^2 + \nonumber  3\delta_1(\epsilon)^2|Z(t)|^2 \\ &\;+ 3\|B\|^2\left(\Omega_1\epsilon + \frac{\delta_1(\epsilon)}{\pi_1^\ast} \right)^2 \|u(t)\|_2^2.
\end{align}
Now following \cite{5373886}, we have that 
\begin{align}
\|w(t)\|_2^2  &\leq \alpha_1(t)\|u(t)\|_2^2  + \alpha_2(t)|Z(t)|^2, \label{eq:w-u-relation} \\
\|u(t)\|_2^2 &\leq \beta_1(t)\|w(t)\|_2^2  + \beta_2(t)|Z(t)|^2, \label{eq:u-w-relation}
\end{align}
where
\begin{align}
\alpha_1(t)  &= 3\left(1 + \int_0^1 \|K e^{A \psi(t) x} B\|^2 \psi(t)^2 dx \right), \\
\alpha_2(t)  &= 3 \int_0^1 \| K e^{A \psi(t) x}\|^2  dx, \\
\beta_1(t)  &= 3\left(1 + \int_0^1 \|K e^{(A+BK) \psi(t)x} B \| ^2 \psi(t)^2 dx \right), \\
\beta_2(t)  &= 3 \int_0^1  \|K e^{(A+BK)\psi(t)x}\|^2 dx\,,
\end{align}
with uniform estimates
\begin{align}
\alpha_1(t)  &\leq \bar{\alpha}_1  = 3\left( 1 + \frac{\|K\|^2 \|B\|^2}{2\pi_1^\ast \|A\|} \left( e^{\frac{2\|A\|}{\pi_1^\ast}} - 1 \right) \right), \\
\alpha_2(t)  &\leq \bar{\alpha}_2  = \frac{3\|K\|^2 \pi_1^\ast}{2\|A\| } \left( e^{\frac{2\|A\|}{\pi_1^\ast}} - 1 \right), \\
\beta_1(t)  &\leq \bar{\beta}_1 = 3\left(1 + \frac{\|K\|^2 \|B\|^2}{2\pi_1^\ast \|A+BK\|}\left( e^{\frac{2\|A+BK\|}{\pi_1^\ast}} - 1 \right)\right), \\
\beta_2(t) &\leq \bar{\beta}_2 = \frac{3\|K\|^2\pi_1^\ast }{2\|A+BK\|}\left( e^{\frac{2\|A+BK\|}{\pi_1^\ast}} - 1 \right).
\end{align}
Hence, we obtain the final perturbation:

\begin{align}
    w^2(1, t) \leq&\; 3|\Omega_1+\delta_1(\epsilon)|^2|\hat{Z}(t) - Z(t)|^2  \nonumber \\ &+ \nonumber  \left[3\delta_1(\epsilon)^2 + 3\bar \beta_2\|B\|^2\left(\Omega_1\epsilon + \frac{\delta_1(\epsilon)}{\pi_1^\ast} \right)^2 \right] |Z(t)|^2 \\ &\;+ 3\bar \beta_1\|B\|^2\left(\Omega_1\epsilon + \frac{\delta_1(\epsilon)}{\pi_1^\ast} \right)^2 \|w(t)\|_2^2.
\end{align}

Notice, the perturbation estimate has error due to both $\hat{\psi}$ and the observer $\hat{Z}(t)$.

\noindent \textbf{\underline{Step 3: A Lyapunov analysis of the full system:}}
Note that, since $A+BK$ and $A-LC$ are Hurwitz, for any $Q, R$ symmetric, positive definite, there exists unique $P, S$ symmetric, positive definite that satisfy the Lyapunov equations
\begin{align}
    P(A+BK)+ (A+BK)^TP = -Q\,, \\
    S(A-LC) + (A-LC)^TS = -R\,. 
\end{align}
Hence, we have that (using Young's inequality)
\begin{align}
    \frac{d}{dt}\left(Z(t)^TPZ(t)\right) =&\; -Z^T(t) QZ(t) + 2Z^T(t)PBw(0, t)\,, \nonumber \\ 
    \leq&\; -\frac{\lambda_{\mathrm{min}}(Q)}{2}|Z(t)|^2 + \frac{2\|PB\|^2}{\lambda_{\mathrm{min}}(Q)}w^2(0, t)\,,
\end{align}
and additionally, letting $\tilde{\xi}(t) = Z(\phi_2(t)) - \xi(t)$, by the observer design, we have
\begin{align}
    \frac{d}{dt}(\tilde{\xi}(t)^T S \tilde{\xi}(t)) \leq&\; -\lambda_{\mathrm{min}}(R) \phi_2'(t)|\tilde{\xi}(t)|^2 \nonumber \\ 
    \leq&\; - \pi_2^\ast\lambda_{\mathrm{min}}(R) |\tilde{\xi}(t)|^2\,. 
\end{align}
Consider the Lyapunov functional:
\begin{align}
    V(t) =&\; Z(t)^T PZ(t) + \frac{4\|PB\|^2}{\pi_1^\ast \lambda_{\text{min}}(Q)}L(t) \nonumber \\ & + \mu \tilde{\xi}(t)^T S \tilde{\xi}(t)\,, 
\end{align}
where $\mu > 0$ is to be given.
Taking the derivative, we obtain
\begin{align}
    \dot{V}(t) \leq&\; -\frac{\lambda_{\mathrm{min}}(Q)}{2}|Z(t)|^2  - \frac{4\|PB\|^2\beta^\ast}{\lambda_{\mathrm{min}}(Q)} L(t)  - \mu\pi_2^\ast\lambda_{\mathrm{min}}(R)|\tilde{\xi}(t)|^2 \nonumber \\ 
    &+ \Omega_2|\Omega_1 + \delta_1(\epsilon)|^2|\hat{Z}(t)-Z(t)| ^2 \nonumber \\ 
    & + \Omega_2 \left( |\delta_1(\epsilon)|^2 +3 \bar{\beta}_2\|B\|^2\left(\Omega_1 \epsilon + \frac{\delta_1(\epsilon)}{\pi_1^\ast}\right) \right) |Z(t)|^2  \nonumber \\ 
    & +  \Omega_2\bar{\beta}_1\|B\|^2\left(\Omega_1\epsilon + \frac{\delta_1(\epsilon)}{\pi_1^\ast} \right)^2 2L(t)\,,  \label{eq:vdot}
\end{align}
where
\begin{align}
    \Omega_2 := \frac{6\|PB\|^2e^{b}}{\pi_0^\ast \pi_1^\ast \pi_2^\ast \lambda_{\mathrm{min}}(Q)}\,. 
\end{align}

If we have that $\epsilon < \epsilon^\ast $ where $\epsilon^\ast$ satisfies the following:
\begin{align}
    \nonumber \max\bigg\{\Omega_2&\Omega_1^2 \left|e^{\|A\|\epsilon^\ast }-1\right|^2,  \Omega_1^2 \left(\epsilon^\ast + \frac{e^{\|A\|\epsilon^\ast}-1}{\pi_1^\ast} \right)^2 \bigg\} \\ &\leq \min\left\{ \frac{\lambda_{\mathrm{min}}(Q)}{12 \Omega_2 \bar{\beta}_2 \|B\|^2}, \frac{2\|PB\|^2 \beta^\ast}{\Omega_2\bar{\beta}_1 \|B\|^2 \lambda_{\mathrm{min}}(Q)}, \frac{\lambda_{\mathrm{min}}(Q)}{4}\right\}\,. 
\end{align}
To give insight into how this condition is derived, we require that the negative terms in \eqref{eq:vdot} dominate the $\epsilon$-dependent perturbation terms. Thus, the above inequality is a sufficient smallness condition ensuring that the approximation error remains below the robustness margin of the closed-loop system. For the engineer, this means that the predictor approximation must be made sufficiently accurate relative to the plant dynamics, controller/observer gains, and delay bounds in order to preserve exponential stability.

\begin{figure*}[ht!]
\centering
    \includegraphics{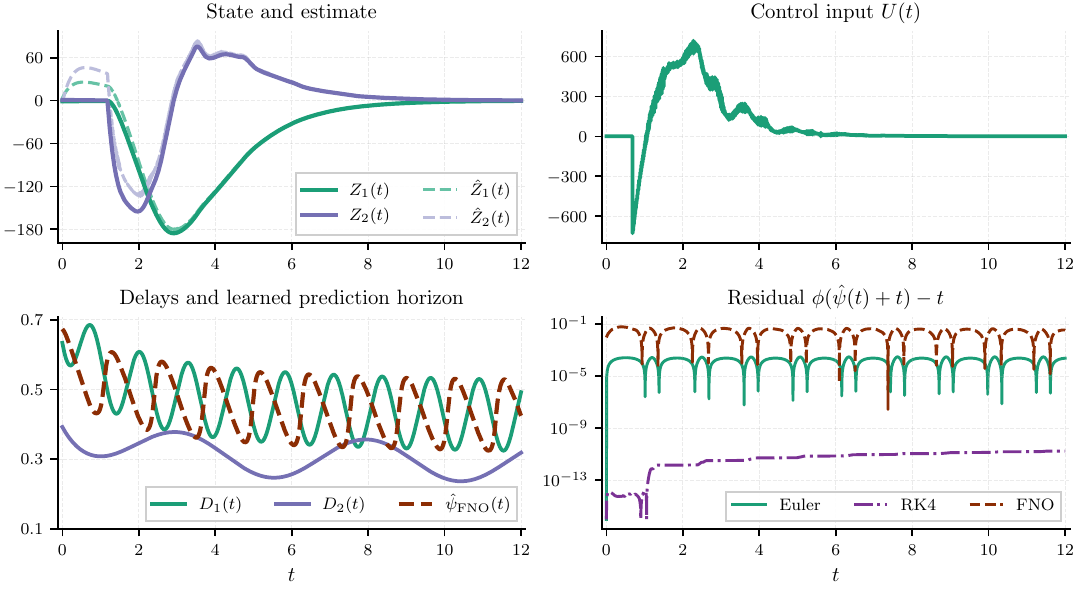}
    \caption{Example of feedback control for the problem in \eqref{eq:numerical-problem-formulation} with the FNO approximation of $\hat{\psi}$. The parameters for the delay function $D_1$ were $(a_1, b_1, \alpha_1, \omega_1, \varphi_1) = (0.4, 0.31, -0.10, 4.95, 0.95)$ and for $D_2$ were $(a_2, b_2, \alpha_2, \omega_2, \varphi_2) = (0.28, 0.15, -0.06, 1.28, 0.82)$. The initial condition of the plant was $Z(0) = [-1, 1]$ and the observer $\xi(0) = [5, -5]$. The input history was fixed to ensure $0$ input for the initial condition and the plant history was set to $Z(0)$ for all times $t \leq 0$.}
    \label{fig:main-figure}
\end{figure*}

Hence, we obtain 
\begin{align}
    \dot{V}(t) \leq&\; -c_1 |Z(t)|^2 - c_2L(t)  \nonumber \\ &\;- \mu\pi_2^\ast\lambda_{\mathrm{min}}(R)|\tilde{\xi}(t)|^2 + \Omega_2|\Omega_1 + \delta_1(\epsilon)|^2|\hat{Z}(t)-Z(t)| ^2
\end{align}
where 
\begin{align}
    c_1 :=&\; \frac{\lambda_{\min}(Q)}{2}
    -\Omega_2 \left( |\delta_1(\epsilon)|^2 +3 \bar{\beta}_2\|B\|^2\left(\Omega_1 \epsilon + \frac{\delta_1(\epsilon)}{\pi_1^\ast}\right)^2 \right)>0, \\
    c_2 :=&\; \frac{4\|PB\|^2\beta^\ast}{\lambda_{\min}(Q)}
    -2\Omega_2\bar{\beta}_1\|B\|^2\left(\Omega_1\epsilon + \frac{\delta_1(\epsilon)}{\pi_1^\ast} \right)^2>0.
\end{align}
To complete, the bound on the Lyapunov function, it requires to bound the observer error. First, note that 
\begin{align}
    |\hat{Z}(t) - Z(t)|^2 =&\; |-e^{A(t-\phi_2(t))}\tilde{\xi}(t)|^2 \nonumber \\ 
    \leq& \; e^{2\|A\|/\pi_1^\ast} |\tilde\xi(t)|^2\,. 
\end{align}
From here, noting that 
\begin{align}
    \Omega_2 \delta_1(\epsilon)^2 \leq \frac{\lambda_{\mathrm{min}}(Q)}{4}\,, 
\end{align}
we can choose $\mu$ as
\begin{align}
    \mu > \frac{\left(2\Omega_2 \Omega_1^2 + \frac{\lambda_{\mathrm{min}}(Q)}{2}\right)e^{2\|A\|/\pi_1^\ast}}{\pi_2^\ast \lambda_{\mathrm{min}}(R)}\,, 
\end{align}
yielding
\begin{align}
    \dot{V}(t) \leq&\; -c_1 |Z(t)|^2 - c_2L(t) -c_3|\tilde{\xi}(t)|^2\,, 
\end{align}
where
\begin{align}
    c_3 := \mu \pi_2^\ast \lambda_{\mathrm{min}}(R) - \Omega_2|\Omega_1+\delta_1(\epsilon)|^2e^{2\|A\|/\pi_1^\ast} \,. 
\end{align}
Notice that $c_1, c_2, c_3$ are all monotonically decreasing in $\epsilon$. Hence, there exists $c_4(\epsilon) > 0$, monotonic, such that 
\begin{align}
    \dot{V}(t) \leq -c_4V(t)\,, \quad \forall t \geq 0
\end{align}
yielding, for all $t \geq 0$
\begin{align}
    V(t) \leq V(0)e^{-c_4t}\,. \label{eq:lyapunov}
\end{align}
From here, using the standard bounds between $\|w\|$ and $\|u\|$ as well as between $\tilde{\xi}$ and $|Z-\hat{Z}|$, we obtain a constant $M > 0$, independent of both time and $\epsilon$, such that
\begin{align}
    \frac{\Gamma(t)}{M} \leq V(t) \leq M\Gamma(t)\,.  \label{eq:Gamma-bnd}
 \end{align}
 From here, using \eqref{eq:lyapunov} and \eqref{eq:Gamma-bnd}, we complete the result. 
\end{proof}
\section{Numerical example} \label{sec:numerics}
\noindent All code for this numerical example is publicly available \cite{bhan2026nap_horizons_code}. 
\textit{Problem formulation:}
To illustrate the approach, we consider the system \eqref{eq:dynamics-1}, \eqref{eq:dynamics-2} with 
\begin{align}
    A = \begin{bmatrix}
        0 & 1  \\ 1 & 2
    \end{bmatrix}\,, \quad B = \begin{bmatrix}
        0 \\ 1 
    \end{bmatrix}\,, \quad C = \begin{bmatrix}
        1 & -1
    \end{bmatrix}\,.  \label{eq:numerical-problem-formulation}  
\end{align}
Notice the open-loop system is unstable and hence we choose $K = [-4, -4]$ and $L=[-4, -8]^T$ as the controller gains to ensure the matrices $A+BK$ and $A-LC$ are Hurwitz. 

\textit{Dataset generation:} In this example, we train a Fourier Neural Operator \cite{FNO:2021} that learns the mapping $D \mapsto \psi$ where $D$ and $\psi$ are given uniformly over $[0, H]$ as discretized function representations. We let $H=12$ in this example. 
The family of delay functions considered is given by
\begin{align}
D(t)
=a
+ \frac{b}{1+t}
+  \alpha \sin \big(\omega t + \varphi\big) \,, \nonumber 
\end{align}
which yields oscillatory delay profiles (See bottom left of Figure \ref{fig:main-figure}). We generate a dataset of 
$2000$ examples by independently sampling the parameters
\begin{align*}
a \sim \mathcal{U}(0.2,\, 3.0)&, \quad 
b \sim \mathcal{U}(0,\, 10), \quad 
\alpha\sim \mathcal{U}(-0.3,\, 0.3), \\
\omega&\sim \mathcal{U}(0.2,\, 3.0), \quad
\varphi \sim \mathcal{U}(0,\, 2\pi),
\end{align*}
where $\mathcal{U}(a, b)$ indicates the uniform distribution over $[a, b]$.
This family of functions are oscillatory and do not have a known analytical inverse, hence they must be approximated. Furthermore, for some choices of the parameters, they do not satisfy the Assumptions (namely $\phi' < 0$). Hence, for each sampled delay function we verify that the assumptions of the theory are satisfied on the interval $[0, H]$ before including the function into our dataset. 

To generate high-fidelity training targets for the neural operator, we compute accurate solutions of the implicit relation
\begin{align}
\phi(\psi(t)+t)-t=0, \nonumber 
\end{align}
which defines $\psi(t)=\phi^{-1}(t)-t$. For each sampled delay function, $\psi(t)$ is computed offline on a fine temporal grid with discretization step $dt=10^{-3}$. At each grid point $t$, the above scalar nonlinear equation is solved to machine precision using a root-finding algorithm, producing highly accurate reference solutions used as supervision for training. 

The dataset of $2k$ pairs is randomly split into 80\% training, 10\% validation, and 10\% test samples. Prior to training, both inputs and outputs are normalized using the empirical dataset mean and standard deviation. We train a one-dimensional Fourier Neural Operator (FNO) \cite{FNO:2021} with $32$ Fourier modes and $64$ hidden channels to learn the mapping $D \mapsto \psi$. The network is trained for $200$ epochs using the Adam optimizer \cite{loshchilov2018decoupled} with learning rate $10^{-3}$, and the training objective is the root-mean-square error between the predicted and reference horizons on the discretized grid. A validation error of $9 \times 10^{-3}$ is achieved during training. Finally, it is important to mention that this computation is performed entirely offline (taking $\approx 3$ hours for dataset generation and training) and therefore does not affect the real-time control implementation. After training, the neural operator provides a direct approximation of the mapping $D \mapsto \psi$. 

\textit{Online deployment:}
During simulation, the trained neural operator is evaluated once on the temporal grid to produce the approximation $\hat{\psi}(t)$, which is then used directly in the predictor-feedback law. This replaces the implicit inversion with a single neural-operator evaluation. For comparison, the Euler-based approach computes the initial condition in \eqref{eq:euler-init-cond-1}, \eqref{eq:euler-init-cond-2} using a root-finding method (bisection method \cite{burden2015numerical}), after which the corresponding ODE is integrated online using the forward Euler method for $t \in [0,T]$. The Runge-Kutta 4 (RK4 \cite{burden2015numerical}) follows the same procedure as Euler, but uses a fourth-order Runge–Kutta (RK4) integrator for time integration.
Figure~\ref{fig:main-figure} shows an example of the closed-loop response using the neural-operator prediction $\hat{\psi}(t)$. The delay functions $D_1$ and $D_2$ lie within the training distribution but were not explicitly seen during training. The observer states converge rapidly and the plant stabilizes. The bottom-right panel reports the consistency error, measuring how well $\hat{\psi}(t)$ satisfies the implicit relation $\phi(\hat{\psi}(t)+t)-t=0$. 

Lastly, to provide a comparison of computational efficiency, we evaluate the wall-clock time required to compute the solution over $1000$ independent samples. On average, the Euler method requires 
$11.98$ ms per evaluation, while RK4 incurs a roughly 
$3.5\times$ higher cost at 
$42.50$ ms. Moreover, the learned FNO surrogate is the fastest due to the GPU usage (Nvidia $4070$), requiring only 
$2.01$ ms per evaluation, corresponding to a 
$6.0\times$ speedup over Euler and
$21.1\times$ over RK4, but comes at a sacrifice of accuracy as shown in the bottom right of Figure \ref{fig:main-figure}. Hence, the frequency of the control input and robustness of the linear system will determine which approximation approach for $\hat{\psi}$ is best.

\section{Conclusion}

In this work, we studied predictor-feedback implementation for linear systems with time-varying input and measurement delays when the inverse delay-time mapping is not available in closed form. We developed an output-feedback design consisting of an observer, a predictor, and the nominal delay-free control law. To address the fact that the predictor horizon depends on the inverse delay-time mapping, we proposed two practical approximation methods: an ODE-based numerical inversion scheme and a neural-operator approximation framework. For both approaches, approximation guarantees were established over compact sets. We then derived sufficient conditions under which exponential closed-loop stability is preserved despite approximation error. Lastly, we presented a numerical example demonstrating the tradeoff between approximation accuracy and computational cost while highlighting the efficacy of the approach. 
\bibliography{references}
\bibliographystyle{elsarticle-num.bst}

\appendix
\section{Proof of Lemma \ref{lem:inverse_lipschitz_in_D}}
\begin{proof}
Fix any $y\in \mathcal{R}$ and define $\psi_i(y) \coloneq \Psi(D_i)(y), \qquad i \in \{1, 2\}$ so that
\begin{align}
    \phi_i^{-1}(y) = y + \psi_i(y)\,. 
\end{align}
Hence, applying $\phi$ to both sides, we have that 
\begin{align}
    y = (y+\psi_i(y)) - D_i(y+\psi_i(y))\,. 
\end{align}

Subtracting yields
\begin{align}
    \psi_1(y)-\psi_2(y) = D_1(y+\psi_1(y)) - D_2(y+\psi_2(y))\,. 
\end{align}
Adding and subtracting $D_2(y+\psi_1(y))$, we have
\begin{align}
    |\psi_1(y) - \psi_2(y)| \leq& |D_1(y+\psi_1(y))-D_2(y+\psi_1(y))| \\ &+ |D_2(y+\psi_1(y))-D_2(y+\psi_2(y))|
\end{align}
Since $\phi_2'(t)=1-D_2'(t)\ge\pi_2^\ast>0$, we have
$D_2'(t)\le 1-\pi_2^\ast$.  
Applying the mean value theorem yields
\begin{align}
|D_2(y+\psi_1(y))&-D_2(y+\psi_2(y))|
\nonumber \\ &\leq
(1-\pi_2^\ast)|\psi_1(y)-\psi_2(y)|.
\end{align}

Hence, combining, we obtain
\begin{align}
    |\psi_1(y)-\psi_2(y)|
\le
\|D_1-D_2\|_{L^\infty}
+
(1-\pi_2^\ast)|\psi_1(y)-\psi_2(y)|\,.
\end{align}
Rearranging yields the direct result. 

\end{proof}

\end{document}